\documentclass[12pt]{article}
\usepackage{amsmath,amssymb,theorem}
\usepackage{graphicx}
  \usepackage{hyperref}

\numberwithin{equation}{section}

\newtheorem{thm}{Theorem}[section]

\newtheorem{prop}[thm]{Proposition}
\newtheorem{cor}[thm]{Corollary}
{\theorembodyfont{\rmfamily}

\newtheorem{rmk}[thm]{Remark}}

\newcommand{\qed}{\hfill \mbox{\raggedright \rule{.07in}{.1in}}}

\newenvironment{proof}{\vspace{1ex}\noindent{\bf
Proof}\hspace{0.5em}}{\hfill\qed\vspace{1ex}}
\newenvironment{pfof}[1]{\vspace{1ex}\noindent{\bf Proof of
#1}\hspace{0.5em}}{\hfill\qed\vspace{1ex}}

\newcommand{\R}{{\mathbb R}}
\newcommand{\C}{{\mathbb C}}
\newcommand{\Z}{{\mathbb Z}}
\newcommand{\D}{{\mathbb D}}

\newcommand{\T}{{\mathbb T}}

 \renewcommand{\Re}{\operatorname{Re}}

\newcommand{\OO}{{\bf O}}
 \newcommand{\E}{{\bf E}}
 \newcommand{\eps}{{\epsilon}}
 \newcommand{\Fix}{\operatorname{Fix}}

\newcommand{\cL}{{\mathcal L}}
\newcommand{\cP}{{\mathcal P}}

\title{Quasicrystals in pattern formation, Part I: Local existence and basic properties} 
\author{
Ian Melbourne \thanks{Mathematics Institute,
University of Warwick,
Coventry, CV4 7AL,
UK}, 
Jens D. M. Rademacher \thanks{
Department of Mathematics,
Universit\"at Hamburg,
20146 Hamburg, Germany}, 
Bob Rink \thanks{
Department of Mathematics,
Vrije Universiteit Amsterdam, 
De Boelelaan 1111, 1081 HV Amsterdam, The Netherlands},
Sergey Zelik \thanks{Department of Mathematics, ZJNU, Jinhua,  China,\\ and Department of Mathematics, University of Surrey, Guildford, UK\\ and Keldysh institute of applied mathematics, Moscow, Russia\\ and HSE University, Nizgnij Novgorod, Russia}.
}

\date{21 October 2024. Updated 1 February 2025}

\begin{document}

\maketitle

\begin{flushright}
{\it Dedicated to the fond memory of Claudia Wulff} \hspace*{3em}
\end{flushright}

 \begin{abstract} 
In this paper, we propose a general mechanism for the existence of quasicrystals in spatially extended systems (partial differential equations with Euclidean symmetry).
We argue that the existence of quasicrystals with higher order rotational symmetry, icosahedral symmetry, etc., is a natural and universal consequence 
of spontaneous symmetry breaking, 
bypassing technical issues such as Diophantine properties and hard implicit function theorems.

The diffraction diagrams associated with these quasicrystal solutions are not Delone sets, so strictly speaking they do not conform to the definition of a ``mathematical quasicrystal''.
But they do appear to capture very well the features of the diffraction diagrams of quasicrystals observed in nature.

For the Swift-Hohenberg equation, we obtain more detailed information, including that the $\ell^2$ norm of the diffraction diagram grows like the square root of the bifurcation parameter.
 \end{abstract}

 \section{Introduction}

The existence of quasicrystals~\cite{SBGC} was first reported in 1984 (see Figure~\ref{fig:experiments}(a)), and such aperiodic order has been the source of much interest ever since. Subsequently, there have been several instances of quasicrystal solutions in fluid experiments. For example,
quasipatterns with eightfold symmetry~\cite{CAL} and
twelvefold symmetry~\cite{EF} were observed in the Faraday wave experiment (see Figure~\ref{fig:experiments}(b)) and quasipatterns with twelvefold symmetry were observed in shaken convection~\cite{VM}.

\begin{figure}
\begin{center}
\begin{tabular}{cc}
\includegraphics[width=0.4\textwidth]{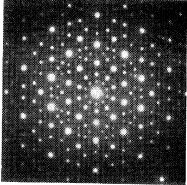} &
\hspace{1em}
\includegraphics[width=0.4\textwidth]{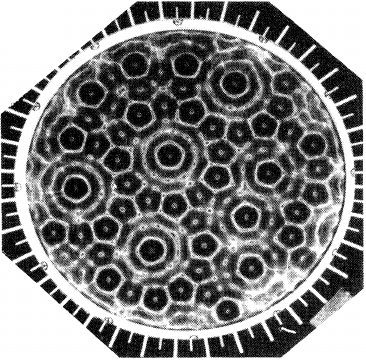} \\
(a) & \hspace{2em} (b)
\end{tabular}
\end{center}
\caption{Experimental evidence of quasicrystal structures.  \newline
(a) Electron diffraction diagram of a metallic solid. See Shechtman \emph{et al.}~\cite{SBGC} for details.
Reprinted Fig.~2 with permission from~\cite{SBGC}.
Copyright (1984) by the American Physical Society. \newline
(b) Photograph of a Faraday wave experiment. See Edwards \& Fauve~\cite{EF} for details.
Reprinted Fig.~3 with permission from \cite{EF}.
Copyright (1993) by the American Physical Society.
}
\label{fig:experiments}
\end{figure}

It is well-known that in systems of partial differential equations (PDEs) with Euclidean $\E(d)$ symmetry, the variation of a parameter generically gives rise to a large variety of spatially periodic solutions~\cite{BCM,CrawKnob91,CrossHoh,DG,GSS88,GS02,Michel80,Satt79}.
This mechanism is known as \emph{spontaneous symmetry breaking}.
It turns out that such bifurcations automatically give rise also to a large class of quasicrystals.

The quasicrystals in this paper are functions 
\begin{align} \label{eq:qc}
u:\R^d\to\R^s, \qquad
u(x)=\sum_{k\in\cL^*}a_ke^{ik\cdot x},
\end{align}
 where $\cL^*$ is a finitely-generated subgroup of $\R^d$
and the amplitudes 
$a_k\in\R^s$ satisfy $\sum_{k\in\cL^*} |a_k|<\infty$.
To such functions $u$, we associate the 
\emph{diffraction diagram}
\[
I_u=\{(k,a_k):a_k\neq0\} \subset\cL^*\times (\R^s\setminus\{0\}).
\]
Also, for each $\eps>0$ we define the \emph{cut-off set}
\[
\Lambda_{u,\eps}=\{k\in\cL^*:|a_k|>\eps\}\subset \cL^*.
\]

\paragraph{Comparison with mathematical quasicrystals}
Mathematical treatments of quasicrystals~\cite{Baake02,Lagarias96,Meyer72,Meyer95,Moody95,Senechal} start from the consideration of a countable set $\Lambda\subset\R^d$ that is \emph{uniformly discrete} ($\inf_{k,\ell\in\Lambda,\;k\neq\ell}|k-\ell|>0$). Usually it is assumed in addition that the set is \emph{relatively dense} (there exists $R\in(0,\infty)$ such that $\bigcup_{x\in\Lambda}B(x,R)=\R^d$).
A set that is uniformly discrete and relatively dense 
is called a \emph{Delone set}~\cite{Delone32}.

For ease of comparison, we add the assumption that the vectors in $\cL^*$ span $\R^d$, but this assumption is not required elsewhere in the paper.
The quasilattice $\cL^*$ is uniformly discrete if and only if it is a lattice, in which case the functions $u$ are spatially periodic and correspond to crystals.
Hence we focus on the case when $\cL^*$ is not uniformly discrete.

The sets $\cL^*$ and $\Lambda_{u,\eps}$ are not Delone sets ($\cL^*$ is relatively dense but not uniformly discrete; $\Lambda_{u,\eps}$ is finite so trivially uniformly discrete  but not relatively dense). 
Nevertheless, the solutions studied in this paper have the desirable properties that
\begin{itemize}
\item[(i)]
There is a finitely-generated relatively dense  subgroup $\cL^*\subset\R^d$ such that
$u(x)=\sum_{k\in\cL^*}a_ke^{ik\cdot x}$ where
the amplitudes
$a_k\in\R^s$ satisfy $\sum_k |a_k|<\infty$;
\item[(ii)] The subgroup of $\cL^*$ generated by $\Lambda_{u,\eps}$ is not uniformly discrete for some $\eps>0$.
\end{itemize}
We conjecture that our solutions typically satisfy the following condition which strengthens condition (ii):
\begin{itemize}
\item[(iii)] For any $M,r>0$ there exists $\eps>0$ such that
$\Lambda_{u,\eps}$ is $r$-dense in $\cL^* \cap B(0,M)$.
\end{itemize}
Consequently, the diffraction diagram $I_u$ ``looks'' like a Delone set to any desired resolution.

Condition~(i) is more restrictive than mathematical definitions of quasicrystals referred to above, and conditions~(ii) and (iii) are more relaxed. They seem to us to better describe the features of naturally arising quasicrystals observed so far, see Figure~\ref{fig:experiments}(a).
For the purposes of this paper, we will say that a function $u:\R^d\to\R^s$ is a quasicrystal if conditions~(i) and~(ii) are satisfied, leaving condition~(iii) for future work.

\begin{rmk} It is possible to relax the assumption that 
	$\sum_{k\in\cL^*} |a_k|<\infty$ and consider classes of almost periodic functions. This is done in the second part~\cite{PartII} of this paper which is from now on referred to as Part~II.
However, the stronger assumption fits naturally into the general framework for bifurcation with Euclidean symmetry set out in~\cite{M99}, see Section~\ref{sec:local}.
\end{rmk}

\paragraph{Existence of quasicrystals via spontaneous symmetry breaking}

Recall that the Euclidean group $\E(d)$ can be regarded as a semidirect product of 
$\OO(d)$ (rotations and reflections) and $\R^d$ (translations).
Let $H$ be a finite subgroup of $\OO(d)$ and fix a unit vector
$k_0\in\R^d$. Define $\cL^*_H$ to be the quasilattice generated by
the vectors $\gamma k_0$, $\gamma\in H$. We require further that $H$ is the largest subgroup of $\OO(d)$ that preserves $\cL^*_H$ (in particular, $-I\in H$).
Such a subgroup $H$ is called a \emph{holohedry}.
We then restrict to functions $u$ that are $H$-invariant, so
the amplitudes $a_k$ satisfy $a_{\gamma k}=a_k$ for all $\gamma\in H$, $k\in\cL^*_H$ (Proposition~\ref{prop:H}).

As explained in Section~\ref{sec:qc}, solutions of this type are guaranteed via spontaneous symmetry breaking. In cases where the holohedry $H$ violates the \emph{crystallographic restriction} (e.g.\ $q$-fold rotations with $q\ge 8$ in $\R^2$, or icosahedral symmetry in $\R^3$) these solutions are quasicrystals satisfying conditions (i) and (ii) above.

Unlike for spatially periodic solutions, we have limited information on the structure of the quasicrystal solutions that we obtain. For example, we do not analyse their dynamics or their stability, and we do not prove existence of smooth branches of solutions. More precise analyses are pursued in~\cite{Iooss13,Iooss14,Iooss17,Iooss19,IoossRucklidge10,RucklidgeRucklidge03}. (See also~\cite{MNT} for nonrigorous results.) However, our aim here is to prove general results without requiring hard implicit function theorems that have to be modified for different equations or in higher dimensions. On the other hand, our results are strengthened when there is information regarding global existence of solutions and estimates of global attractors.

\begin{rmk}
In traditional approaches to spontaneous symmetry breaking and equivariant bifurcation theory~\cite{CrawKnob91,GSS88,Michel80,Ruelle73,Satt79,Vand82}, there is the concept of \emph{isotropy subgroup} $H$ of the full group of symmetries $\Gamma$.  
The phase space $V$ of the underlying equations is $\Gamma$-invariant, and the 
\emph{fixed-point subspace} $\Fix H$ is defined to be the subspace 
 $\Fix H=\{v\in V:hv=v\text{ for all }h\in H\}$. 
Such a fixed-point subspace is
automatically a flow-invariant subspace for the underlying equations. Restricting to $\Fix H$ results in an evolution equation whose solutions have at least the symmetries in the isotropy subgroup $H$.

The quasicrystals discussed here have isotropy subgroup $H\subset\E(d)$, but $\Fix H$ contains many solutions that are not spatially quasiperiodic.
For instance, a reasonable choice of $V$ is the space of continuous bounded functions. 
If $d=2$ and $H=\D_8$ say, then there are plenty of octagonally-symmetric functions in $V$ that have no approximate translation symmetry.\footnote{For spatially-periodic square solutions, say, this issue does not arise since we can take the isotropy subgroup to include both $\D_4$ and the appropriate subgroup $\Z^2$ of the translations.}
Even if the $\E(2)$-invariant phase space $V$ is restricted \emph{a priori} to quasiperiodic functions (which is artificial since there are plenty of other interesting solutions such as fronts and spirals that occur via spontaneous symmetry breaking), $\Fix H$ will contain solutions that do not satisfy condition~(i).

Even though there do not exist isotropy subgroups corresponding to quasicrystal structure, it turns out that Euclidean-equivariance always implies the existence of flow-invariant subspaces, that are not fixed point subspaces, consisting entirely of quasicrystal solutions and spatially constant solutions (Corollary~\ref{cor:equivariant}). 
Moreover, it is easy to arrange instability of the spatially constant solutions, and hence the only possible stable dynamics within these subspaces consists of quasicrystal solutions.
It is in this sense that quasicrystals are universal in systems with Euclidean symmetry.
\end{rmk}

The remainder of Part~I of this paper is organised as follows.
In Section~\ref{sec:qc}, we describe the flow-invariant subspaces of quasicrystal solutions that arise naturally and universally in Euclidean-equivariant systems of PDEs via spontaneous symmetry breaking.
In Section~\ref{sec:SH}, we analyse in detail a simplified model for pattern formation, namely the Swift-Hohenberg equation~\cite{SwiftHoh}.
In Section~\ref{sec:Bruss}, we briefly consider reaction-diffusion equations; specifically the Brusselator.

\section{Flow-invariant subspaces of quasicrystals}
\label{sec:qc}

In this section, we define spaces of planar quasicrystals (Subsection~\ref{sec:qcplanar}) and higher-dimensional quasicrystals (Subsection~\ref{sec:qchigher}).
In Subsection~\ref{sec:local}, we discuss local existence and uniqueness results for Euclidean-equivariant PDEs. In particular, we  show that the spaces of quasicrystals are flow-invariant and hence give rise to quasicrystals via spontaneous symmetry breaking.
In Subsection~\ref{sec:global}, we briefly discuss global existence results, referring to Part~II for details.

\subsection{Planar quasicrystals with holohedry $\D_q$}
\label{sec:qcplanar}

For simplicity, we focus first on the case of Euclidean equivariant PDEs in the plane. The phase space of such PDEs consists of functions $u:\R^2\times\Omega\to\R^s$ where $\Omega$ is bounded\footnote{For example, in the case of planar Rayleigh-B\'enard convection (governed by the Boussinesq equations), $\Omega=[0,1]$.}. By~\cite{M99}, the analysis of bifurcations with nonzero critical wavenumber generically reduces to the case where $\Omega=\{0\}$, so we consider functions $u:\R^2\to\R^s$. 
(By~\cite{M99}, we could also assume $s=1$, but there is no gain here in doing so and the reduction to $s=1$ fails in higher dimensions.)
We suppose in addition that the action of the Euclidean group $\E(2)$ is {\em scalar}, so 
$(\gamma\cdot u)(x)=u(\gamma^{-1}x)$ for all $\gamma\in\E(2)$.
(There is also the \emph{pseudoscalar} action~\cite{BCM,GolubitskyShiauTorok03} where reflections act as 
$(\gamma\cdot u)(x)=-u(\gamma^{-1}x)$, see Remark~\ref{rmk:anti} below.
By~\cite{M99}, generically these are the only two possibilities.)

Let $q\ge2$ be an even integer and
let $\D_q\subset\OO(2)$ denote the dihedral group of order $2q$ generated by rotations of order $q$ and a reflection.
We say that $u$ is $\D_q$-invariant if $\gamma\cdot u=u$ for all $\gamma\in\D_q$. Since we assume the scalar action, this means that $u(\gamma x)=u(x)$ for all
$\gamma\in \D_q$, $x\in\R^2$.

Fix a vector $k_0\in\R^2$ with $|k_0|=1$ and
 let $\cL^*_{\D_q}$ be the subgroup of $\R^2$ generated by $\{\gamma k_0:\gamma\in\D_q\}$.
Abusing notation slightly, we let $\ell^1(\cL^*_{\D_q})$ consist of $\D_q$-invariant functions $u:\R^2\to\R^s$ of the form
\[
u(x)=\sum_{k\in\cL^*_{\D_q}}a_ke^{ik\cdot x}, \quad a_k\in\R^s,
\]
such that $\sum_{k\in\cL^*_{\D_q}}|a_k|<\infty$.
By Proposition~\ref{prop:H} below, $\D_q$-invariance means that
$a_{\gamma k}=a_k$ for all $k\in\cL^*_{\D_q}$,
$\gamma\in\D_q$.

\begin{rmk} Usually, one considers complex amplitudes $a_k\in\C^s$ satisfying the reality condition
$a_{-k}=\overline{a_k}$. However,
since $q$ is even, $-I\in\D_q$ and hence $\overline{a_k}=a_{-k}=a_k$, so automatically $a_k\in\R^s$ for all $k$. 
\end{rmk}

We say that $u$ is \emph{spatially constant} if $u(x)\equiv c$ for some $c\in\R^s$.
Such solutions are $\E(2)$-invariant and are contained in $\ell^1(\cL^*_{\D_q})$.
For $q\le 6$, all elements of $\ell^1(\cL^*_{\D_q})$ have nontrivial translation invariance.  
However, for $q\ge8$, non-spatially constant elements of $\ell^1(\cL^*_{\D_q})$ automatically satisfy conditions (i) and (ii) from the Introduction and hence are quasicrystals with holohedry $\D_q$. Moreover, such functions typically satisfy also condition~(iii).

Given a function $u\in\ell^1(\cL^*_{\D_q})$, we define the norm $\|u\|_1=\sum_{k\in\cL^*_{\D_q}}|a_k|$.
We define also the (incomplete) $\ell^2$ norm
$\|u\|_2=\big(\sum_{k\in\cL^*_{\D_q}}a_k^2\big)^{1/2}$ which is finite for
$u\in \ell^1(\cL^*_{\D_q})$.

Let $C_b(\R^2)$ denote the Banach space of continuous bounded functions $u:\R^2\to\R^s$ with the supnorm $\|u\|_\infty=\sup_{x\in\R^2}|u(x)|$.
We have the embedding $\ell^1(\cL^*_{\D_q})\subset C_b(\R^2)$ for each $q$, with
$\|u\|_\infty\le \|u\|_1$.\footnote{Throughout Part I, $\|\;\|_1$ and $\|\;\|_2$ denote the $\ell^1$ and $\ell^2$ norms on sequences, and
$\|\;\|_\infty$ denotes the supremum norm on functions.}

\begin{rmk}[Superquasicrystals]
If $k_0$ does not lie on a reflection axis, 
then $\{\gamma k_0:\gamma\in\D_q\}$ consists of $2q$ elements.
For each $q\ge8$, the quasicrystals with holohedry $\D_q$  lie on a four-dimensional family of solutions.
Three of these dimensions are due to the Euclidean symmetry. Modulo symmetry, there is a one-dimensional family parametrised by the smallest angle between $k_0$ and its images under reflection.

When $q\le 6$, the lattices where $k_0$ is not fixed by a reflection in $\D_q$ are often referred to as ``superlattices'', leading to ``supersquares'' and ``superhexagons''~\cite{AF,DG}.
As far as we know, the corresponding ``superquasicrystal'' solutions have not been observed experimentally, but they automatically exist in systems with Euclidean symmetry by the arguments presented here.
\end{rmk}

\begin{rmk}[Anti-quasicrystals]
\label{rmk:anti}
As mentioned above, there are also pseudoscalar actions where reflections acts as $(\gamma\cdot u)(x)=-u(\gamma^{-1}x)$.
In such systems, we obtain ``anti-quasicrystals'' (or ``pseudoscalar quasicrystals'').
To do this, choose $k_0$ not fixed by a reflection, and consider the space of functions $u(x)=\sum_{k\in\cL^*}a_ke^{ik\cdot x}$ with $a_{\gamma k}=a_k$ when $\gamma$ is a rotation and $a_{\gamma k}=-a_k$ when $\gamma$ is a reflection.

For analogous pseudoscalar spatially periodic solutions (e.g.\ anti-squares), see~\cite{BCM,GolubitskyShiauTorok03}.

The symmetry group of the planar Swift-Hohenberg equation (considered in Section~\ref{sec:SH}) is $\E(2)\times\Z_2$ (Euclidean transformations in the plane and $u\mapsto -u$).
For $q\ge8$, the maximal subgroup of $\E(2)\times\Z_2$ preserving $\ell^1(\cL^*_{\D_q})$ 
is $\D_q\times\Z_2$. Although our  focus is on quasicrystals with holohedry $\D_q$,
we simultaneously obtain anti-quasicrystals with holohedry $\D_q^-$ consisting of (1) $q$-fold rotations and (2) reflections composed with $u\mapsto-u$.
\end{rmk}

\subsection{Higher-dimensional quasicrystals}
\label{sec:qchigher}

The structures described in Subsection~\ref{sec:qcplanar} easily extend to
PDEs with Euclidean symmetry $\E(d)$ in general dimension $d$.
Again we restrict to phase spaces of functions
$u:\R^d\to\R^s$ (this is no loss of generality by~\cite{M99}) and to \emph{scalar} actions of $\E(d)$ where $(\gamma\cdot u)(x)=u(\gamma^{-1}x)$ for $\gamma\in\E(d)$, $x\in\R^d$.

Let $H$ be a finite subgroup of $\OO(d)$ and fix a unit vector $k_0\in\R^d$.
Define
$\cL^*_H$ to be the subgroup of $\R^d$ generated by $\{\gamma k_0:\gamma\in H\}$.
We require that $H$ is the 
maximal subgroup of $\OO(d)$ preserving $\cL^*_H$ so $H$ is a \emph{holohedry}.

Abusing notation slightly, we let $\ell^1(\cL^*_H)$ be the space of functions $u:\R^d\to\R^s$ of $H$-invariant functions of the form
\[
u(x)=\sum_{k\in\cL^*_H}a_ke^{ik\cdot x}, \quad a_k\in\R^s,
\]
such that 
$\|u\|_1=\sum_{k\in\cL^*_H}|a_k|<\infty$.
Since the action of $\E(d)$ is assumed to be scalar, $H$-invariance means that $u(\gamma x)=u(x)$ for all $\gamma \in H$, $x\in\R^d$.

\begin{prop} \label{prop:H}
Let $u\in \ell^1(\cL^*_H)$ as above. Then $u$ is $H$-invariant if and only if 
$a_{\gamma k}=a_k$ for all $\gamma\in H$, $k\in\cL^*_H$.
\end{prop}

\begin{proof}
Recall that $H\subset\OO(d)$ and so acts orthogonally on vectors $k\in\R^d$.
Since $H$ preserves $\cL^*_H$,
\[
u(\gamma x)=
\sum_{k\in\cL^*_H}a_k e^{ik\cdot (\gamma x)}
= \sum_{k\in\cL^*_H}a_k e^{i(\gamma^{-1}k)\cdot x}
= \sum_{k\in\cL^*_H}a_{\gamma k} e^{ik\cdot x}
\quad\text{for all $x\in\R^d$, $\gamma\in H$.}
\]
The result follows by equating amplitudes.
\end{proof}

Again, we have the embedding $\ell^1(\cL^*_H)\subset C_b(\R^d)$ with
$\|u\|_\infty\le \|u\|_1$.

As in Subsection~\ref{sec:qcplanar}, we are particularly interested in the case when $\cL^*_H$ is not uniformly discrete. Then functions in $\ell^1(\cL^*_H)$ are either spatially constant or are quasicrystals satisfying conditions~(i) and~(ii) (and typically~(iii)) in the Introduction.

\subsection{Local well-posedness and local existence of quasicrystal solutions}
\label{sec:local}

Let $\cP(\cL^*_H)$ denote the subspace of $H$-invariant trigonometric polynomials (finite sums) within $\ell^1(\cL^*_H)$.

\begin{prop} \label{prop:equivariant}
Let $N:\big(\cP(\cL^*_H)\big)^r\to \{u:\R^d\to\R^s\}$ be an $r$-linear operator.
Suppose that $N$
satisfies the equivariance condition
\[
N(\gamma\cdot u_1,\dots,\gamma\cdot u_r)=\gamma\cdot N(u_1,\dots,u_r)
\quad\text{for all $u_1,\dots,u_r\in\cP(\cL^*_H)$, $\gamma\in\E(d)$.}
\]
Then writing $u_j(x)=\sum_{k\in\cL^*_H}a_{k,j}e^{ik\cdot x}$,
there exists $S:(\R^d)^r\to\C^s$ such that
\[
N(u_1,\dots,u_r)(x)=\sum_{k_1,\dots,k_r\in\cL^*_H}S(k_1,\dots,k_r)a_{k_1,1}\cdots a_{k_r,r}
e^{i(k_1+\dots+k_r)\cdot x}, \quad x\in\R^d.
\]
In particular, $N$ maps $\big(\cP(\cL^*_H)\big)^r$ into $\cP(\cL^*_H)$.
\end{prop}

\begin{proof}
Fix $k_1,\dots ,k_r\in\R^d$ and set
$v=N(u_1,\dots,u_r)$ where $u_j=e^{ik_j\cdot x}$.
Taking $\gamma$ to be translation by $a\in\R^d$, it follows from equivariance and $r$-linearity that
\begin{align*}
v(x-a) & =N(e^{-ik_1\cdot a}u_1,\dots,e^{-ik_r\cdot a}u_r)(x)
\\ & =e^{-i(k_1+\dots+k_r)\cdot a}N(u_1,\dots,u_r)(x)
=e^{-i(k_1+\dots+k_r)\cdot a}v(x).
\end{align*}
Taking $a=x$, we obtain $v(x)=
e^{i(k_1+\dots+k_r)\cdot x}S$ 
where $S=v(0)$.
Hence $N(u_1,\dots,u_r)(x)=S(k_1,\dots,k_r)
e^{i(k_1+\dots+k_r)\cdot x}$.
The result follows by $r$-linearity.~
\end{proof}

\begin{rmk} In the proof we used equivariance of $N$ under translation symmetries. There are additional restrictions on the symbol $S$ resulting from $\OO(d)$-equivariance and since the functions $u$ are $\R^s$-valued.
In the linear case $r=1$, it is easily seen that
$S(k)=S_0(|k|)$ for some $S_0:[0,\infty)\to\R^s$.
\end{rmk}

Now define $\ell^1(\R^d)$ to be the Banach space of functions $u:\R^d\to\R^s$ of the form
\[
u(x)=\sum_{k\in\R^d}a_ke^{ik\cdot x}, \quad a_k\in\C^s,\; a_{-k}=\overline{a_k}
\]
with finite norm
$\|u\|_1=\sum_{k\in\R^d}|a_k|<\infty$.

\begin{cor} \label{cor:equivariant}
Let $N:\big(\ell^1(\R^d)\big)^r\to \ell^1(\R^d)$  be a  continuous $\E(d)$-equivariant $r$-linear operator.
Then $N$ restricts to
$N:\big(\ell^1(\cL^*_H)\big)^r\to \ell(\cL^*_H)$ for each holohedry $H$.
\end{cor}

\begin{proof} This is immediate from Proposition~\ref{prop:equivariant}.
\end{proof}

We are interested in ``semilinear parabolic'' PDEs of the form
\[
u_t=F(u)=Lu+N_2(u,u)+\dots+N_r(u,\dots,u)
\]
where $L$ and $N_2,\dots,N_r$ are $\E(d)$-equivariant multilinear operators but are in general unbounded.

A general framework for bifurcations with Euclidean symmetry is laid out in~\cite{M99}.
The function space there consists of regular Borel measures on $\R^d$ and 
in particular includes the subspace $\ell^1(\R^d)$.
For a large class of $\E(d)$-equivariant PDEs (including the Boussinesq equations as well as the examples considered in this paper) it can be shown that the linear part $L$ is a sectorial operator~\cite{Henry} on $C_b(\R^d)$ and on $\ell^1(\R^d)$, and that the PDE $u_t=F(u)$  defines a local dynamical system on these spaces.
Smoothing ensures that
solutions $u(t)$ with $u(0)\in\ell^1(\R^d)$ 
satisfy $(Lu)(t)\in\ell^1(\R^d)$ for $t>0$ small.
The fact that the spaces $C_b(\R^d)$ and $\ell^1(\R^d)$ are Banach algebras makes it particularly easy to deal with the nonlinear terms $N_j$, $j\ge2$.
See~\cite{M98,M99} for previous use of such $\ell^1$ spaces, as well 
as~\cite{Giga05,GigaMahalovYoneda11} and references therein.

\begin{rmk} \label{rmk:equivariant}
It follows as in the proof of Proposition~\ref{prop:equivariant} that the subspaces $\ell^1(\cL^*_H)$ are flow-invariant (for $t>0$ small) for these local dynamical systems.
\end{rmk}

Now we introduce a bifurcation parameter $\lambda\in\R$ and consider $\E(d)$-equivariant PDEs
\[
u_t=F(u,\lambda)=L_\lambda u+\sum_{j=2}^r N_{j,\lambda}(u,\dots,u).
\]
There is a ``trivial'' spatially constant steady-state solution $u\equiv0$.
It is assumed that inside $\ell^1(\R^d)$ and $C_b(\R^d)$, the trivial solution is a sink for $\lambda<0$ (the eigenvalues $\mu$ of $L_\lambda$ satisfy $\sup\Re\mu<0$) and linearly unstable for $\lambda>0$ small ($L_\lambda$ possesses eigenvalues with positive real part). Moreover, we suppose that $\ker L_0$ is spanned by eigenfunctions of the form $b_ke^{ik\cdot x}$ with $|k|=k_c$ where the \emph{critical wavenumber} $k_c$ is positive. 
This is called a \emph{bifurcation of type \emph{I}$_s$} in~\cite{CrossHoh} and is called a 
\emph{steady-state bifurcation with nonzero critical wavenumber} in~\cite{M99}.

Restricting to flow-invariant subspaces $\ell^1(\cL^*_H)$ with $\cL^*_H$ not uniformly discrete, we see that
for $\lambda>0$ small, spatially constant solutions are unstable and all remaining solutions are quasicrystals on their intervals of existence.
This scenario applies in particular to the Boussinesq equations.
Hence the existence of quasicrystals as defined in the Introduction are a natural and universal consequence of the symmetry of the equations.

In this generality, the quasicrystal solutions exist for a finite amount of time
uniformly in $\lambda>0$ small, but may blow up in finite time, or they may exist for all time but with diverging norm as $t\to\infty$.
In subsequent sections, we show that more can be said in specific examples.

\begin{rmk} \label{rmk:AP}
Define the space $AP(\cL^*_H)$ of \emph{almost periodic functions} to consist
	of functions $u:\R^d\to\R^s$ arising as uniform limits of trigonometric polynomials with frequencies in $\cL^*_H$.
We have the inclusions 
\[
\ell^1(\cL^*_H)\subset AP(\cL^*_H)\subset C_b(\R^d).
\]
	Arguments based on the theory of almost-periodic functions~\cite{Bes54,LeZhi} imply 
that $\|u\|_2\le \|u\|_\infty$ for all functions $u\in AP(\cL^*_H)$ and
	can be used to obtain flow-invariance of $AP(\cL^*_H)$.
	See Part~II for details.
\end{rmk}

\subsection{Global existence of quasicrystal solutions}
\label{sec:global}

Global well-posedness is known for numerous $\E(d)$-equivariant PDEs~\cite{MiranvilleZelik}.
Moreover, it is often possible to show that solutions $u(t)$ with $u(0)\in C_b(\R^d)$ satisfy $\sup_{t\ge0}\|u(t)\|_\infty{<\infty}$.
In such situations, we are able to obtain sharper results as illustrated in the examples in the remainder of this paper.
For instance, by Remark~\ref{rmk:AP}, global existence and boundedness in the $\ell^2$ norm for solutions in $AP(\cL^*_H)$ is an immediate consequence of these properties in $C_b(\R^d)$.

Global existence and boundedness in the $\ell^1$ norm is not given by existing theory.
However, by standard arguments it is often the case that if the amplitudes of $u(0)$ decay sufficiently quickly, then $u(t)\in\ell^1(\cL^*_H)$ for all $t\ge0$ (though it may still be the case that $\|u(t)\|_1$ is unbounded).

We refer to Part~II for details regarding the statements in this subsection.

\section{The Swift-Hohenberg equation}
\label{sec:SH}

A simple example of a Euclidean-equivariant PDE is
the Swift-Hohenberg equation~\cite{SwiftHoh} given by
\begin{align} \label{eq:SH}
\partial_t u=F(u,\lambda)=-(\Delta+1)^2u+\lambda u-u^3.
\end{align}
Here, the phase space consists of functions
$u:\R^d\to\R$ 
and $\lambda\in\R$ is a parameter.

There is a trivial spatially constant solution $u\equiv0$ which loses stability as
$\lambda$ passes through zero. 
The linearisation $(dF)_{0,0}$ has a zero eigenvalue with kernel consisting of wavefunctions $e^{ik\cdot x}$ with $|k|=1$.
Hence there is 
a steady-state bifurcation with nonzero critical wavenumber.

As mentioned in Section~\ref{sec:local},
it is easily seen that these equations are locally well-posed on $C_b(\R^d)$  and also on 
$\ell^1(\cL^*_H)$ for all holohedries $H$.
In fact, it is well-known for sufficiently low-dimensional Swift-Hohenberg 
equations that we have global existence and boundedness of solutions in 
$C_b(\R^d)$, see~\cite{MiranvilleZelik,PartII}. 
The techniques in~\cite{MiranvilleZelik,PartII} work provided $d\le 9$ and
we restrict to this situation throughout this section.
Under this assumption, 
we show how to obtain global families of quasicrystal solutions with $\sqrt\lambda$-growth in the $\ell^2$ norm.
(As explained in Section~\ref{sec:local}, we still have local existence of quasicrystals even when $d\ge10$.)

By Remark~\ref{rmk:AP}, it follows that we have global existence and boundedness of solutions in $AP(\cL^*_H)$.
As mentioned in Section~\ref{sec:global}, a further standard argument in 
Part~II shows that initial conditions for which the amplitudes $a_k$ decay sufficiently quickly also have global existence in $\ell^1(\cL^*_H)$ but without control on the norm.

Our main result for the Swift-Hohenberg equation gives the existence of quasicrystals and upper bounds as well as lower bounds in the $\ell^2$ norm.

\begin{thm} \label{thm:SH}
Let $d\le 9$ and $H\subset\OO(d)$ be a holohedry.  
\begin{itemize}
\item[(a)] 
For $\lambda\le 0$, the trivial solution $u\equiv0$ is globally asymptotically stable in $AP(\cL^*_H)$ with the $\ell^2$ norm,
and the convergence is exponential for $\lambda<0$.
\item[(b)]
Let $\lambda>0$, $\eps>0$.
For any initial condition $u_0\in AP(\cL^*_H)$, there exists $t_0\ge0$ such that the solution $u(t)$ with $u(0)=u_0$ 
satisfies
$\|u(t)\|_2\le(1+\eps)\sqrt\lambda$ for all $t\ge t_0$.
\item[(c)]
There is a family of solutions
$u_\lambda(t)\in\ell^1(\cL^*_H)$ defined for all $\lambda>0$, $t\ge0$, and a constant $C_H\in(0,1]$ such that for $\lambda>0$, $t>0$,
\[
\|u_\lambda(t)\|_2\ge C_H \sqrt\lambda
\quad\text{and}\quad
\|(\Delta+1)u_\lambda(t)\|_2^2+\tfrac12\|u_\lambda(t)^2\|_2^2\le \lambda^2.
\]
\item[(d)]
For $\lambda\in(0,1)$, the solutions in (c) are bounded away from spatially constant functions: 
$\inf_{t\ge0}\inf_{c\in\R}\|u_\lambda(t)-c\|_2>0$.
\end{itemize}
\end{thm}

\begin{rmk}
The Swift-Hohenberg equation has a variational structure which we exploit in the proof of Theorem~\ref{thm:SH}. However, we are unable to say anything about the dynamics of our quasicrystal solutions. In particular, we do not claim existence of steady-state quasicrystal solutions.
(We note that it seems possible to exploit the variational structure and methods of~\cite{Zelik03} to obtain $H$-invariant steady-state solutions for all holohedries $H$.)

In addition, we do not discuss the asymptotic stability properties of the solutions in Theorem~\ref{thm:SH}.
\end{rmk}

\begin{rmk}
The condition $\lambda\le1$ is required in Theorem~\ref{thm:SH}(d) since there are two further spatially constant solutions 
$u(t)\equiv \pm \sqrt{\lambda-1}$
for $\lambda>1$.
\end{rmk}

\begin{rmk} \label{rmk:d=2} We refer to Part~II for further estimates when $d=2$.
There, it is shown that the trivial solution is globally asymptotically stable in $C_b(\R^2)$ for $\lambda{\le0}$ and hence (by upper-semicontinuity of the global attractor) there is a function $\alpha:(0,\infty)\to(0,\infty)$ with
$\lim_{\lambda\to0}\alpha(\lambda)=0$ such that solutions in $C_b(\R^2)$ for $\lambda>0$ satisfy $\|u(t)\|_\infty\le \alpha(\lambda)$ for $t$ sufficiently large.
(The corresponding result for $d\ge3$ is unknown and
there is no estimate on $\alpha(\lambda)$ even for $d=2$.)
\end{rmk}

\begin{rmk} 
The solutions $u_\lambda(t)$ in Theorem~\ref{thm:SH}(c,d) satisfy conditions~(i) and~(ii) from the Introduction, and hence are quasicrystals, whenever $\cL^*_H$ is not uniformly discrete.
As shown in Part II, a slight modification produces solutions that also satisfy condition~(iii).
First, perturb the initial condition $u_\lambda(0)=\sqrt\lambda\sum_{k\in\cL^*_H}a_k(0) e^{ik\cdot x}$ so that $a_k(0)\neq0$ for all $k$.
Analyticity properties can be used to show that the amplitudes $a_k(t)$ vanish only for isolated values of $t$ for each $k\in\cL^*_H$.
Hence $u_\lambda(t)$ satisfies condition~(iii) for all but countably many values of $t$.
\end{rmk}

\paragraph{Asymptotic persistence of quasicrystal structure}
Write the solutions $u_\lambda(t)$ in Theorem~\ref{thm:SH}(c) as
\[
u_\lambda(t)=\sum_{k\in\cL^*_H}a_{k,\lambda}(t) e^{ik\cdot x}.
\]
Fix $\lambda>0$.
The lower bound in Theorem~\ref{thm:SH}(c)
still admits the possibility that $\lim_{t\to\infty}\sum_{k\in\cL^*_H,\,|k|\le M}|a_{k,\lambda}(t)|^2=0$ for all $M>0$. 
If this were the case, then the quasicrystal nature of the diffraction diagram would only be a transient on every bounded subset of $\R^d$.
This scenario is excluded by our next result.
For $u\in \ell^1(\cL^*_H)$, 
$u(x)=\sum_{k\in\cL^*_H}a_k e^{ik\cdot x}$, 
write 
\[
u=V_r^u+E_r^u, \qquad  
V_r^u(x)=\sum_{k\in\cL^*_H, \,||k|^2-1|<r}a_k e^{ik\cdot x}, \quad
E_r^u(x)=\sum_{k\in\cL^*_H, \,||k|^2-1|\ge r}a_k e^{ik\cdot x}.
\]

\begin{thm} \label{thm:SH2}
Let $u_\lambda(t)$ be the family of solutions in Theorem~\ref{thm:SH}
and let $t\ge0$, 
$0<\epsilon<1$, $\lambda>0$, $r>0$ such that $\lambda \le \eps r^2$. 
Then
$\|V_r^{u_\lambda(t)}\|_2^2\ge (1-\eps) \|u_\lambda(t)\|_2^2$.
\end{thm}

\begin{rmk} \label{rmk:SH2}
Fixing $\eps\in(0,1)$, and $M>1$ large, we see by
Theorem~\ref{thm:SH2} in conjunction with the lower bound in Theorem~\ref{thm:SH}(c) that 
\[
\sum_{k\in\cL^*_H,\,|k|\le M}|a_{k,\lambda}(t)|^2
\ge \|V_{M^2+1}^{u_\lambda(t)}\|_2^2
\ge (1-\eps)\|u_\lambda(t)\|_2^2
\ge(1-\eps) C_H^2\lambda
\]
for all $\lambda\le \eps (M^2+1)^2$ and all $t\ge0$.
Hence, energy does not flow to high frequency modes ($|k|$ large) as $t$ increases, and this property holds uniformly on compact subsets of parameter space. 

Similarly, for any $\delta\in(0,1)$, it follows that
$\inf_{t\ge 0}\sum_{k\in\cL^*_H,\,|k|\ge \delta}|a_{k,\lambda}(t)|^2>0$ for all $\lambda>0$ sufficiently small.
Hence energy does not flow to low frequency modes ($k\approx0$). 

Moreover, for any $\lambda<1$, we can choose $r\in(\lambda^{1/2},1)$ and
$\eps=\lambda r^{-2}\in(0,1)$. Hence we obtain Theorem~\ref{thm:SH}(d) as a consequence of Theorem~\ref{thm:SH2}.
\end{rmk}

In the remainder of this subsection, we prove Theorems~\ref{thm:SH}(a)--(c)
and~\ref{thm:SH2}.
Define the (incomplete)
inner product by $\langle u,v\rangle=\sum_{k\in\cL^*_H} a_kb_k$
for functions
$u(x)=\sum_ka_ke^{ik\cdot x}$,
$v(x)=\sum_kb_ke^{ik\cdot x}\in \ell^1(\cL^*_H)$.

\begin{prop} \label{prop:inner}
Let $u,v,w\in\ell^1(\cL^*_H)$.\footnote{Note that $\|u\|_2<\infty$ for $u\in\ell^1$ and $uv\in \ell^1$ for $u,v\in\ell^1$, so all expressions in this proposition are well-defined.} Then
\begin{itemize}
\item[(a)]  
$\langle u,vw \rangle = \langle uv,w \rangle$.
\item[(b)] $\langle u,u\rangle^2\le \langle u^2,u^2\rangle$.
\item[(c)] 
$\langle (\Delta+1)^2u,v\rangle = 
\langle (\Delta+1)u,(\Delta+1)v\rangle$
whenever $(\Delta+1)^2u\in\ell^1(\cL^*_H)$, $(\Delta+1)v\in\ell^1(\cL^*_H)$.
\end{itemize}
\end{prop}

\begin{proof}
Write
	$u(x)=\sum_ka_ke^{ik\cdot x}$,
	$v(x)=\sum_kb_ke^{ik\cdot x}$,
	$w(x)=\sum_kc_ke^{ik\cdot x}$.
Then
$vw=\sum_k\bigl(\sum_mb_mc_{k-m}\bigr)
e^{ik\cdot x}$, and hence
$\langle u,vw \rangle = 
\sum_k\bigl(a_k\sum_m b_m c_{k-m}\bigr)$.
Similarly,
$\langle uv,w \rangle = 
\sum_k\bigl(\sum_mb_ma_{k-m} c_k\bigr)$.
Changing $m$ to $-m$ and using that $b_{-m}= b_m$, we obtain
$\langle uv,w \rangle = 
\sum_k\bigl(\sum_m b_ma_{k+m} c_k\bigr)$.
Changing $k$ to $k-m$ yields~(a).

For~(b), note that
$\langle u,u\rangle^2=(\sum_m a_m^2)^2$
is the ``$k=0$'' term in the sum 
$\langle u^2,u^2\rangle =
\sum_k(\sum_m a_ma_{k-m})^2$.

Finally,
\begin{align*}
\langle (\Delta+1)^2u,v\rangle = 
\sum (-|k|^2+1)^2a_k\; b_k &=
\sum (-|k|^2+1)a_k\;(-|k|^2+1) b_k 
\\ & =
\langle (\Delta+1)u,(\Delta+1)v\rangle,
\end{align*}
proving (c).
\end{proof}

\begin{rmk} \label{rmk:iso}
	In Part~II, we define $\ell^2(\cL^*_H)$ as an appropriate completion.
Let $k_1,\dots,k_p$ be a minimal set of generators for $\cL^*_H$.
  There is an isometric isomorphism between $\ell^2(\cL^*_H)$ and the space $L^2(\T^p)$ given by
\[
\iota: \sum_{m\in\Z^p}a_{m_1,\dots,m_p}e^{i(m_1k_1+\dots +m_pk_p)\cdot x}
\longmapsto
 \sum_{m\in\Z^p}a_{m_1,\dots,m_p} 
e^{2\pi i(m_1\theta_1+\dots +m_p\theta_p)}.
\]
With this identification, part (a) of Proposition~\ref{prop:inner} is immediate
(writing $U=\iota(u)$ and so on, we get $\int_{\T^p} U(VW)=\int_{\T^p} (UV)W$) and (b) follows from the
Cauchy-Schwarz inequality ($(\int_{\T^p} U^2)^2=(\int_{\T^p} U^21)^2\le \int_{\T^p} U^4$).
However, we caution that the
action of $\Delta$ is not the standard one, so part (c) requires
an extra calculation.
\end{rmk}

\begin{pfof}{Theorem~\ref{thm:SH}(a,b)}
Define $N=\|u\|_2^2=\langle u,u\rangle$.
Using Proposition~\ref{prop:inner}, we compute that
\begin{align*}
\tfrac12 N_t & =\langle u,\partial_t u\rangle = \langle u,-(\Delta+1)^2u\rangle
+\langle u,\lambda u\rangle + \langle u,-u^3 \rangle \\
& = -\langle (\Delta+1)u,(\Delta+1)u\rangle
+\lambda \langle u,u\rangle - \langle u^2,u^2 \rangle \\
& \le \lambda   \langle u,u\rangle -  \langle u,u\rangle^2 =N(\lambda - N).
\end{align*}
In particular, 
$N_t\le 2\lambda N$, so 
$N(u(t))\le e^{2\lambda t}N(u(0))$.
Hence if $\lambda<0$, then 
all trajectories converge to $0$ exponentially quickly.

If $\lambda=0$, then $N_t\le -2N^2$ from which it follows that
$N(u(t))\le N(u(0))/(2N(u(0))t+1)$ and again all trajectories
converge to $0$.

If $\lambda>0$, then 
using
$N_t  \le 2N(\lambda - N)$,
we see that
$N$ decreases exponentially along trajectories
outside the ball of radius $(1+\eps)^2\lambda$ for all $\eps>0$.
\end{pfof}

Define the potential
\[
P_\lambda(u)=\tfrac12\|(\Delta+1)u\|_2^2-\tfrac12\lambda \|u\|_2^2
+\tfrac14 \|u^2\|_2^2.
\]

\begin{prop}  \label{prop:P}
Let $u(t)$ be a solution to the Swift-Hohenberg equation~\eqref{eq:SH} for some $\lambda\in\R$ and suppose that $u(0)\in \ell^1(\cL^*_H)$.
Then 
\[
\frac{d}{dt}P_\lambda(u(t))=-\|F(u(t),\lambda)\|_2^2
\quad\text{for all $t>0$.}
\]
\end{prop}

\begin{proof}  
Let $v=\partial_t u$.
By Proposition~\ref{prop:inner}(a,c) and the smoothing properties of the PDE,
\begin{align*}
(dP_\lambda)_{u}v & =\langle (\Delta+1)u,(\Delta+1)v\rangle -\lambda\langle u,v \rangle
+\langle u^2,uv\rangle
\\ & =\langle (\Delta+1)^2u,v\rangle -\lambda\langle u,v \rangle
+\langle u^3,v\rangle
=-\langle F(u,\lambda),v\rangle.
\end{align*}
It follows as usual that 
\(
\frac{d}{dt}P_\lambda(u(t))=(dP_\lambda)_{u(t)}\partial_tu(t)=-\langle F(u(t),\lambda),F(u(t),\lambda)\rangle.
\)
\end{proof}

\begin{rmk} \label{rmk:min}
(a)
When $\lambda\le0$, the potential function $P$ has a unique global minimum at $u=0$.  

\vspace{1ex}
\noindent (b)
By Proposition~\ref{prop:inner}(b),
$P_\lambda(u)\ge \frac14 \bigl(\|u\|_2^4-2 \lambda\|u\|_2^2 \bigr)$.
In particular, $P_\lambda(u)\ge -\frac14\lambda^2$ for all $u$,~$\lambda$.
\end{rmk}

\begin{pfof}{Theorem~\ref{thm:SH}(c)}
Choose $k_0\in\R^d$ with $|k_0|=1$ such that $|\{\gamma k_0:\gamma\in H\}|=|H|$.
(The case where $k_0$ is fixed by an element of $H$ is almost identical.)
For $\lambda\ge0$, consider the initial condition 
\[
u_\lambda(0)\in \ell^1(\cL^*_H), \qquad 
u_\lambda(0)(x)=a\sqrt\lambda \sum_{\gamma\in H} e^{i\gamma k_0\cdot x},
\]
where $a>0$.  Since this is a finite sum, the amplitudes certainly decay rapidly and hence $u_\lambda(t)\in\ell^1(\cL^*_H)$ for all $t\ge0$.

Clearly, 
$\|(\Delta+1)^2u_\lambda(0)\|_2=0$ and $\|u_\lambda(0)\|_2^2=|H| a^2\lambda$.
A rough estimate gives $a^4\lambda^2|H|^2\le \|u_\lambda(0)^2\|_2^2\le a^4\lambda^2|H|^4$.
(Note that $u_\lambda(0)^2$ is a sum of $|H|^2$ vectors in $\R^d$ (including multiplicities), so $\|u_\lambda(0)^2\|_2^2=a^4\lambda^2(M_1^2+\dots+M_p^2)$ where
$p$ is the number of distinct vectors in the sum and $M_1,\dots,M_p$ are the multiplicities. 
The extreme cases are $p=|H|^2$, $M_1=\dots=M_p=1$ and
$p=1$, $M_1=|H|^2$.) Hence we can write 
$\|u_\lambda(0)^2\|_2^2 = a^4\lambda^2|H|^2g_H$
where $1\le g_H\le |H|^2$.
Altogether,
\[
P_\lambda(u_\lambda(0))  =-\tfrac12 |H|a^2\lambda^2+\tfrac14 |H|^2g_Ha^4\lambda^2
= \tfrac14\lambda^2 |H|(|H|g_Ha^4-2a^2).
\]
This is minimised at
$a=(|H|g_H)^{-1/2}$, and for this choice of $a$ we obtain
\[
P_\lambda(u_\lambda(0))=-\tfrac{1}{4g_H}\lambda^2.
\]
By Proposition~\ref{prop:P},
\begin{equation} \label{eq:P}
P_\lambda(u_\lambda(t))\le -\tfrac{1}{4g_H}\lambda^2
\quad\text{for all $t\ge0$.}
\end{equation} 
Writing again $N=\|u_\lambda(t)\|_2^2$ and 
using Remark~\ref{rmk:min}(b),
\[
\tfrac14(N^2-2\lambda N)\le P_\lambda(u_\lambda(t))\le -\tfrac{1}{4g_H}\lambda^2.
\]
It follows that $(N-\lambda)^2\le (1-g_H^{-1})\lambda^2$
and hence
\[
N\ge C_H^2\lambda \quad\text{where} \quad
C_H^2= 1-\sqrt{1-g_H^{-1}}\in(0,1].
\]
This gives the lower bound in Theorem~\ref{thm:SH}(c).

Using~\eqref{eq:P} once more, 
\[
\|(\Delta+1)u_\lambda(t)\|_2^2-\lambda\|u_\lambda(t)\|_2^2
+\tfrac12\|u_\lambda(t)^2\|_2^2= 2P_\lambda(u_\lambda(t))\le -\tfrac{1}{2g_H}\lambda^2.
\]
By~(b),
\[
\|(\Delta+1)u_\lambda(t)\|_2^2 +\tfrac12\|u_\lambda(t)^2\|_2^2 \le \big((1+\eps)^2-\tfrac{1}{2g_H}\big)\lambda^2,
\]
so choosing $\eps$ small enough yields the upper bound in Theorem~\ref{thm:SH}(c).
\end{pfof}

\begin{rmk}
For $d=2$, the holohedries are $H=\D_q$ where $q$ is even. 
In this case, $|H|=2q$.
As before, 
$u_\lambda(0)^2$ is a sum of $|H|^2$ vectors in $\R^2$ (including multiplicities).
The vector $0$ occurs with multiplicity $|H|$ and there are 
$|H|$ distinct vectors of length two (each with multiplicity $1$). The remaining $|H|^2-2|H|$ vectors group into pairs
and comprise $\frac12(|H|^2-2|H|)$  distinct vectors of multiplicity $2$.
Hence 
\[
\|u_\lambda(0)^2\|_2^2
=a^4\lambda^2(1\cdot|H|^2+|H|\cdot 1^2+\tfrac12(|H|^2-2|H|)2^2)=(3|H|^2-3|H|)a^4\lambda^2.
\]
This gives $g_H=(3|H|-3)/|H|$, so
\[
C_H^2=1-\sqrt{\frac{2|H|-3}{3|H|-3}}
=1-\sqrt{\frac{4q-3}{6q-3}}.
\]
\end{rmk}

Now we turn to the proof of Theorem~\ref{thm:SH2}.

\begin{prop}\label{prop:E}
Let $0<\epsilon<1$, $\lambda\ge0$, $r>0$ such that $\lambda \le \eps r^2$. 
Suppose that $u,\,\Delta u\in\ell^1(\cL^*_H)$.
If $\|V_r^u\|_2^2 \leq (1-\eps) \|u\|_2^2$, then $P_{\lambda}(u)\geq 0$.
\end{prop}

\begin{proof}
Note that 
\(
\langle V_r^u,E_r^u\rangle=
\langle \Delta V_r^u,\Delta E_r^u\rangle=0,
\)
and so 
\[
\|u\|_2^2=\|V_r^u\|_2^2 + \|E_r^u\|_2^2, \qquad \|(1+\Delta)u\|_2^2 = \|(1+\Delta)V_r^u\|^2_2 + \|(1+\Delta)E_r^u\|^2_2.
\]
In particular, the assumption $\|V_r^u\|_2^2\leq (1-\eps)\|u\|_2^2$ is equivalent to 
$\|E_r^u\|_2^2 \geq \eps \|u\|_2^2$. This together with the 
assumption $\lambda \le \eps r^2$ implies that 
\[
 (r^2-\lambda) \| E_r^u\|^2_2 - \lambda \|V_r^u\|^2_2 
 =  r^2 \| E_r^u\|^2_2 - \lambda \|u\|_2^2    \geq (\eps r^2 -\lambda )\|u\|_2^2  \geq 0 .
\]
In addition,
\[
\|(\Delta+1)E_r^u\|^2_2 =\sum_{k\in\cL^*_H,\,||k|^2-1|>r} (|k|^2-1)^2a_k^2 \geq r^2 \|E_r^u\|_2^2.
\]
Hence
\begin{align*}
2P_\lambda(u) \geq&  \|(\Delta+1)(V_r^u+E_r^u)\|_2^2 
- \lambda \|V_r^u+E_r^u\|_2^2\\
\geq & 
\|(\Delta+1) E_r^u \|_2^2  - \lambda \|V_r^u\|_2^2 - \lambda \|E_r^u\|_2^2 
\geq   (r^2-\lambda)\|E_r^u\|_2^2 -  \lambda \|V_r^u\|_2^2 \geq 0 . 
 \end{align*}
 This proves the result.
\end{proof}

\begin{pfof}{Theorem~\ref{thm:SH2}}
Let $t\ge0$.
By~\eqref{eq:P}, $P_\lambda(u_\lambda(t))<0$.
Hence the result follows from Proposition~\ref{prop:E}.
\end{pfof}

\section{Reaction diffusion systems and Turing instabilities}
\label{sec:Bruss}

Steady-state bifurcations with nonzero critical wavenumber akin to that in the Swift-Hohenberg equation \eqref{eq:SH} frequently occur in reaction diffusion systems, where these are usually referred to as Turing instabilities \cite{Tur}. Prototypical are two component systems, and a prominent example is the so-called Brusselator \cite{PL68}. Its simple standard form is given by the equations
\[
\begin{aligned}\label{e:Brusselator}
u_t &= d_1 \Delta u + A - (B+1)u+ u^2 v,\\
v_t &= d_2\Delta v +Bu - u^2v,
\end{aligned}
\]
where $(u,v):\R^d\to\R^2$ and
$d_1, d_2, A,B$ are positive parameters. 
The system has been studied broadly from a physical viewpoint, in particular deriving conditions for instability and bifurcations of simple patterns, e.g.\ \cite{Verdasca}. 

There is a trivial spatially constant solution $(u,v)\equiv (A,B/A)$, which plays the role of the zero state in the Swift-Hohenberg equation. Let $\eta=\sqrt{d_1/d_2}$. We assume that $1<\eta<((1+A^2)-1)/A$. 
Fix $A>0$ and let $\lambda_\mathrm{B}=B-(1+A\eta)^2$.

It is readily computed that the trivial solution undergoes a Turing instability as $\lambda_\mathrm{B}$ passes through $0$. 
For $\lambda_\mathrm{B}<0$ the trivial solution is linearly stable,
for $\lambda_\mathrm{B}>0$ it is linearly unstable,
 and for $\lambda_\mathrm{B}=0$
the kernel of the linearised equations at $(u,v)\equiv (A,B/A)$
consists of Fourier-modes $(u_0,v_0)e^{i k\cdot x}$ for $|k| = \sqrt{A/\eta}$, where $(u_0,v_0)$ is a multiple of 
$(-A,A\eta^2+\eta)$. 
Hence, there is a steady-state bifurcation with Euclidean symmetry and nonzero critical wavenumber. 

It is easily seen that this PDE defines a local dynamical system in
$C_b(\R^d)$ and in the $\ell^1$ spaces as in Section~\ref{sec:local}.
Hence the trivial solution $(u,v)\equiv (A,B/A)$ is locally asymptotically stable for $\lambda<0$ and unstable for $\lambda>0$, and
we obtain local existence of quasicrystals by the discussion in Section~\ref{sec:local}.
Moreover, global existence of solutions in $C_b(\R^d)$, and hence in each $AP(\cL^*_H)$, is proved in Part~II so we obtain quasicrystal solutions $u(t)$ defined for all $t$.

However, there are notable differences to the Swift-Hohenberg equation \eqref{eq:SH}. First, the Brusselator has no gradient structure and e.g.\ admits time periodic spatially periodic solutions \cite{Verdasca}. Second, while bifurcations in \eqref{eq:SH} are supercritical, this is generally not the case in the Brusselator (analogous to the variant of \eqref{eq:SH} with an additional quadratic term). 
Certain spatially periodic steady-state solutions (such as hexagons when $d=2$) exist for both $\lambda_\mathrm{B}>0$ and $\lambda_\mathrm{B}<0$ small. In particular, $(u,v)\equiv (A,B/A)$ is not globally asymptotically stable for $\lambda_\mathrm{B}< 0$, so the analogue of Theorem~\ref{thm:SH}(a) fails.
We have not pursued $\ell^2$-estimates like those in Theorem~\ref{thm:SH}(b,c,d).


\begin{thebibliography}{19}
\bibitem{AF} H. Arbell and J. Fineberg. Pattern formation in two-frequency forced parametric wave. \emph{Phys. Rev. E} \textbf{65} (2002) 036224.

\bibitem{Baake02} M. Baake. A guide to mathematical quasicrystals.
In: Quasicrystals eds.\ J-B. Suck, M. Schreiber and P. H\"{a}ussler (Berlin: Springer) (2002) pp.\ 17--48.

\bibitem{Bes54}
A. Besicovitch. \emph{Almost Periodic Functions}. Dover, New York, 1954.

\bibitem{BCM}
I.~{Bosch Vivancos}, P.~Chossat and I.~Melbourne. {New planforms in systems of
  partial differential equations with Euclidean symmetry}. \emph{Arch. Rational
  Mech. Anal.} \textbf{131} (1995) 199--224.
 
\bibitem{Iooss13}
B.~Braaksma, G.~Iooss and L.~Stolovitch. Existence of quasipattern solutions
  of the {S}wift-{H}ohenberg equation. \emph{Arch. Ration. Mech. Anal.}
  \textbf{209} (2013) 255--285.

\bibitem{Iooss14}
B.~Braaksma, G.~Iooss and L.~Stolovitch. Erratum to: {E}xistence of
  quasipattern solutions of the {S}wift-{H}ohenberg equation.
  \emph{Arch. Ration. Mech. Anal.} \textbf{211} (2014) 1065.

\bibitem{Iooss17}
B.~Braaksma, G.~Iooss and L.~Stolovitch. Proof of quasipatterns for the
  {S}wift-{H}ohenberg equation. \emph{Comm. Math. Phys.} \textbf{353} (2017)
  37--67.

\bibitem{Iooss19}
B.~Braaksma and G.~Iooss. Existence of bifurcating quasipatterns in steady
  {B}\'{e}nard-{R}ayleigh convection. \emph{Arch. Ration. Mech. Anal.}
  \textbf{231} (2019) 1917--1981.

 \bibitem{CAL} B. Christiansen, P. Alstr\o m and M. T. Levinsen.
 Ordered capillary-wave states: quasicrystals, hexagons and radial waves.
 \emph{Phys. Rev. Lett.} \textbf{68} (1992) 2157--60.

\bibitem{CrawKnob91}
J.~D. Crawford and E.~Knobloch. {Symmetry and symmetry-breaking bifurcations in
  fluid dynamics}. \emph{Annu. Rev. Fluid Mech.} \textbf{21} (1991) 341--387.

\bibitem{CrossHoh}
M.~C. Cross and P.~C. Hohenberg. {Pattern formation outside of equilibrium}.
  \emph{Rev. of Mod. Phys.} \textbf{65} (1993) 851--1112.

\bibitem{Delone32} B. N. Delone. Neue Darstellung der geometrischen Kristallographie.
\emph{Z. Kristallographie} \textbf{84} (1932) 109--149.

 \bibitem{DG} B. Dionne and M. Golubitsky.
 Planforms in two and three dimensions.
 \emph{Z. Angew. Math. Phys.} \textbf{43} (1992) 36--62.

\bibitem{EF} W. S. Edwards and S. Fauve.
Parametrically excited quasicrystal surface waves.
\emph{Phys. Rev. E} \textbf{47} (1993) 788--91.

 \bibitem{Giga05}
 Y.~Giga, K.~Inui, A.~Mahalov and S.~Matsui. Uniform local solvability for the
   {N}avier-{S}tokes equations with the {C}oriolis force. \emph{Methods Appl.
   Anal.} \textbf{12} (2005) 381--393.

\bibitem{GigaMahalovYoneda11}
Y.~Giga, A.~Mahalov and T.~Yoneda. On a bound for amplitudes of
  {N}avier-{S}tokes flow with almost periodic initial data. \emph{J. Math.
  Fluid Mech.} \textbf{13} (2011) 459--467.

\bibitem{GS02}
M.~Golubitsky and I.~N. Stewart. \emph{{The Symmetry Perspective}}. Progress in
  Mathematics \textbf{200}, {Birkh\"auser}, Basel, 2002.

\bibitem{GSS88}
M.~Golubitsky, I.~N. Stewart and D.~Schaeffer. \emph{{Singularities and Groups
  in Bifurcation Theory, Vol. II}}. Appl. Math. Sci. \textbf{69}, Springer, New
  York, 1988.

\bibitem{GolubitskyShiauTorok03}
M. Golubitsky, L. J. Shiau and A. T\"{o}r\"{o}k. Bifurcation on the
  visual cortex with weakly anisotropic lateral coupling. \emph{SIAM J. Appl.
  Dyn. Syst.} \textbf{2} (2003) 97--143.

\bibitem{Henry}
D.~Henry. \emph{{Geometric Theory of Semilinear Parabolic Equations}}. Lecture
  Notes in Math., Springer, Berlin, 1981.


\bibitem{IoossRucklidge10}
G.~Iooss and A.~M. Rucklidge. On the existence of quasipattern solutions of the
  {S}wift-{H}ohenberg equation. \emph{J. Nonlinear Sci.} \textbf{20} (2010)
  361--394.


\bibitem{Lagarias96}
J.~C. Lagarias. Meyer's concept of quasicrystal and quasiregular sets.
  \emph{Comm. Math. Phys.} \textbf{179} (1996) 365--376.

\bibitem{LeZhi}
B. Levitan and V. Zhikov. \emph{Almost periodic functions and differential equations}.
Cambridge University Press, 1982.

\bibitem{MNT} B. A. Malomed, A. A. Nepomnyashchi\u i and M. I. Tribelski\u i.
Two-dimensional quasiperiodic structures in nonequilibrium systems.
\emph{Sov. Phys. JETP} \textbf{69} (1989) 388--96.

\bibitem{MiranvilleZelik} A. Miranville and S. Zelik. Attractors for dissipative partial differential equations in bounded and
unbounded domains. Handbook of differential equations: evolutionary equations. Handb. Differ. Equ.,
Elsevier/North-Holland, Amsterdam IV (2008), 103--200.

\bibitem{M98}
I.~Melbourne. {Derivation of the time-dependent Ginzburg-Landau equation on the
  line}. \emph{J. Nonlinear Sci.} \textbf{8} (1998) 1--15.

\bibitem{M99}
I.~Melbourne. {Steady state bifurcation with Euclidean symmetry}. \emph{Trans.
  Amer. Math. Soc.} \textbf{391} (1999) 1575--1603.

\bibitem{PartII} I. Melbourne, J. D. M. Rademacher, B. Rink and S.~Zelik.
Quasicrystals in pattern formation. Part II: Spatially almost periodic profiles and global existence. Preprint, January 2025. arXiv:2501.18042.

\bibitem{Meyer72} Y. Meyer.  \emph{Algebraic Numbers and Harmonic Analysis}.
North-Holland, Amsterdam, 1972.
 
\bibitem{Meyer95} Y. Meyer.  Quasicrystals, Diophantine approximation and 
algebraic numbers. In: \emph{Beyond Quasicrystals}, eds.\ F. Axel and D. Gratias. Springer, Berlin, 1995, pp. 3--16.

\bibitem{Michel80}
L.~Michel. {Symmetry defects and broken symmetry. Configurations. Hidden
  symmetry}. \emph{Rev. of Mod. Phys.} \textbf{52} (1980) 617--651.

\bibitem{Moody95} R. V. Moody. Meyer sets and the finite generation of quasicrystals. In: \emph{Symmetries in Science VIII}, ed.\ B. Gruber. Plenum, New York, 1995, pp.\ 379--394.

\bibitem{PL68}
I.~Prigogine and R.~Lefever. Symmetry breaking instabilities in dissipative systems. II. 
\emph{J. Chem. Phys.} \textbf{48} (1968) 1695--1700.


\bibitem{RucklidgeRucklidge03}
A.~M. Rucklidge and W.~J. Rucklidge. Convergence properties of the 8, 10 and 12
  mode representations of quasipatterns. \emph{Phys. D} \textbf{178} (2003)
  62--82.

\bibitem{Ruelle73}
D.~Ruelle. {Bifurcations in the presence of a symmetry group}. \emph{Arch.
  Rational Mech. Anal.} \textbf{51} (1973) 136--152.

\bibitem{Satt79}
D.~Sattinger. \emph{{Group Theoretic Methods in Bifurcation Theory}}. Lecture
  Notes in Math. \textbf{762}, Springer, Berlin, 1979.

\bibitem{SBGC} D. Shechtman, I. Blech, D. Gratias and J. W. Cahn.
Metallic phase with long-range orientational order and no translational symmetry.
\emph{Phys. Rev. Lett.} \textbf{53} (1984) 1951--1954.

\bibitem{Senechal}
M.~Senechal. \emph{Quasicrystals and geometry}. Cambridge University Press,
  Cambridge, 1995.

\bibitem{SwiftHoh}
J.~B. Swift and P.~C. Hohenberg. {Hydrodynamic fluctuations at the convective
  instability}. \emph{Phys. Rev. A} \textbf{15} (1977) 319--328.

\bibitem{Tur} A.~M. Turing. The chemical basis of morphogenesis. 
Phil. Trans. R. Soc. Lond. B \textbf{237} (1952) 37--72.

\bibitem{Vand82}
A.~Vanderbauwhede. \emph{{Local Bifurcation and Symmetry}}. Pitman Research
  Notes in Math. \textbf{75} Boston, 1982.

\bibitem{Verdasca}
J. Verdasca, A. {de Wit}, G. Dewel and P. Borckmans. Reentrant hexagonal Turing structures.
\emph{Phys. Lett. A} \textbf{168} (1992) 194--198.

\bibitem{VM} U. E. Volmar and H. W. M\"{u}ller. 
Quasiperiodic patterns in Rayleigh–B\'{e}nard convection under gravity modulation.  \emph{Phys. Rev. E} \textbf{56} (1997) 5423--5430.

\bibitem{Zelik03}
S.~Zelik. Formally gradient reaction-diffusion systems in {$\mathbb R^n$} have
  zero spatio-temporal topological entropy, 2003. Dynamical systems and
  differential equations (Wilmington, NC, 2002), pp.~960--966.

\end{thebibliography}
\end{document}